\documentclass[letterpaper, 10 pt, conference]{ieeeconf}  

\IEEEoverridecommandlockouts                              

\usepackage{cite}
\usepackage{amsmath}
\usepackage{amssymb}
\usepackage{enumerate}
\usepackage{url}
\usepackage{cite}

\usepackage{amsthm}
\usepackage{txfonts}
\usepackage{ifpdf}
\ifpdf
   \usepackage[pdftex]{graphicx}
   \usepackage{epstopdf}
   \AppendGraphicsExtensions{.gif}
\else
   \usepackage[dvips]{graphicx}
\fi

\theoremstyle{plain}
\newtheorem{lem}{Lemma}
\newtheorem{corr}[lem]{Corollary}
\newtheorem{thm}{Theorem}

\theoremstyle{definition}
\newtheorem{defn}{Definition}

\theoremstyle{remark}

\title{Cooperative Content Offloading Through WiFi and Mobile Device-to-Device Networks}

\author{Guoqiang Mao$^{1,2}$ and Xiaofeng Tao$^{2}$
	\thanks{This research is supported by Chinese National Science Foundation project 61428102.}
	\thanks{$^{1}$ School of Computing and Communications, The University of Technology Sydney}%
	\thanks{$^{2}$ National Engineering Lab. for Mobile Network Security, Beijing University of Posts and Telecommunications}%
}

\begin{document}

\maketitle
\thispagestyle{empty}
\pagestyle{empty}

\begin{abstract}
This paper investigates the use of WiFi and mobile device-to-device networks, with vehicular ad hoc networks being a typical example, as a complementary means to offload and reduce the traffic load of cellular networks. A novel cooperative content dissemination strategy is proposed for a heterogeneous network consisting of different types of devices with different levels of mobility, ranging from static WiFi access points to mobile devices such as vehicles. The proposed strategy offloads a significant proportion of data traffic from cellular networks to WiFi or device-to-device networks. Detailed analysis is provided for the content dissemination process in heterogeneous networks adopting the strategy. On that basis, the optimal parameter settings for the content dissemination strategy are discussed. Simulation and numerical studies show that the proposed strategy effectively reduces data traffic for cellular networks while guaranteeing successful content delivery.
\end{abstract}

\section{Introduction}\label{section_introduction}
Recent years have seen an exponential annual increase in  mobile data traffic\cite{li-multiple-2013, andrews-seven-2013}. Content dissemination  and offloading, where a portion of data traffic is offloaded from cellular networks to high-capacity and low-cost complementary networks, including WiFi, vehicular and device-to-device networks, to ease the burden of cellular networks, becomes  an increasingly important and challenging task \cite{zhang-enabling-2013, lee-mobile-2013, Kannan08Robust, Ge15Energy, Mao09Road, Mao09Graph, Mao06Online}.

In addition to complementary networks, the mobility of users can also be exploited to assist the content dissemination. When users carry mobile devices physically while walking around university campus, shopping centres  or travelling by taxis, buses or private vehicles, the content in their mobile devices also move with them without consuming any bandwidth. Together with device-to-device communication technologies such as vehicle-to-vehicle communications, mobility of users offers an alternative way to transport delay-tolerant content efficiently and cost effectively~\cite{li-multiple-2013}.

A large proportion of the content delivered over mobile networks is delay-tolerant content, like videos, newspapers, weather reports and vehicular info-entertainment. For these types of content, content offloading reduces the traffic load of cellular networks (and boosts the capacity of cellular networks in the sense that more users can be served) and also provides a higher data rate for the users. The expense is that content offloading will incur higher delay compared with direct transmission using cellular networks.

This paper proposes a novel cooperative content dissemination and offloading strategy for a heterogeneous network consisting of different types of devices with different levels of mobility, ranging from static WiFi access points (APs) to mobile devices such as vehicles, to reduce the traffic load of cellular networks while guaranteeing the successful delivery of content. The strategy is particularly suited for the delivery of delay tolerant content.
More specifically, the following contributions are made in the paper: 
\begin{enumerate}
\item a cooperative content dissemination strategy is proposed, which exploits all three data dissemination methods, i.e. cellular networks, complementary networks and device mobility, to disseminate content; 
\item erasure codes, to be introduced in Section \ref{section_system_model}, are employed to further reduce the data traffic load;
\item analytical results are presented to stochastically characterise the content dissemination process, taking into account the heterogeneity in the devices, in terms of mobility and transmission capability. In particular, the reduction in the amount of data traffic in cellular networks is calculated;
\item based on the above results, optimal parameter settings of the content dissemination strategy are discussed, which minimise the data traffic load of the cellular network while guaranteeing the successful delivery of all content.
\end{enumerate}

The rest of this paper is organised as follows: Section~\ref{section_related_work} reviews related work. Section~\ref{section_system_model} introduces the system model, including the content dissemination strategy. 
Section~\ref{section_analysis} presents the analysis of the content dissemination process, whose optimal design is studied in Section \ref{section_optimization}. Section~\ref{section_simulation} validates the analysis using simulations. Finally Section~\ref{section_conclusion} concludes this paper and discusses possible future work.

\section{Related work}\label{section_related_work}

In this section, we give a brief review of the work most related to this paper.

Recent research \cite{lee-mobile-2013} has shown that WiFi networks have already carried and offloaded a large amount of mobile data traffic. When a device enters into a WiFi covered area, it can switch its data traffic from cellular networks to WiFi networks \cite{bulut-wifi-2013} to reduce the traffic load of cellular networks. A major issue in WiFi offloading is the optimum deployment of WiFi APs. Bulut et al. \cite{bulut-wifi-2013} compared different methods of deploying WiFi APs for efficient offloading of mobile data traffic. They also proposed a greedy approach that can achieve a high offloading efficiency. 

In addition to using WiFi APs, recent technology has also allowed mobile devices to become a virtual WiFi access points (a.k.a. WiFi Tethering), so that devices can communicate with one another via ad hoc connections without relying on infrastructure. In this way, devices can \emph{cooperate} with one another to disseminate content. 

Using a typical cooperative content dissemination strategy \cite{aijaz-a-2013}, the service providers first deliver the content to only a small group of users, then these users can further disseminate the content to other subscribed users when their mobile devices are in the proximity and can communicate using WiFi tethering or Bluetooth technology. It is obvious that such opportunistic content dissemination cannot guarantee the delivery of content. This paper proposes a mechanism that can provide guarantee on the delivery of content.

Furthermore, some work in cooperative mobile data offloading (e.g. \cite{li-multiple-2013,bo-mobile-2012}) investigated the design and selection of helpers, viz. some special mobile devices selected to help the content provider to deliver messages to other mobile devices using ad hoc connections. 
%

As vehicular networks form an important class of mobile networks, there is significant research on content dissemination in vehicular networks \cite{zhang-enabling-2013}. 
Wireless access through vehicle-to-roadside communications can be used in public transport vehicles for streaming applications, e.g., videos and interactive advertisements. As pointed out in \cite{niyato-a-2010}, it is challenging to develop an efficient wireless access scheme that minimises the cost of wireless connectivity for downloading data. There are also some empirical studies on content dissemination methods for vehicular networks \cite{panichpapiboon-a-2012}.

\section{System model and content dissemination strategy}\label{section_system_model}
\subsection{Network model}\label{section_network_model}
We consider a network  of $N$ nodes. These $N$ nodes can be classified into $H$ types according to their mobility, e.g. static WiFi APs, mobile devices carried by pedestrians and vehicles, and social characteristics, e.g. students in the same class. Let $N_h$ be the number of nodes of the $h^{th}$ type:  $\sum_{h=1}^H N_h=N$.

Suppose that at some initial time instant $t=0$, the $N$ nodes are independently, randomly and uniformly distributed on a torus $(0,L]^2$ \cite{franceschetti-random-2007}. 
It follows that the nodes' density in this network is $\lambda=N/L^2$. 
Then these nodes start to move independently according to  some mobility models\cite{camp-a-2002}.  We assume that the nodes' mobility is such that the spatial distribution of nodes is stationary and ergodic with stationary uniform distribution on the torus. As shown in  \cite{nain-properties-2005}, a number of mobility models have this property.

The use of a toroidal rather than planar region as a tool in analysing network properties is well known. The use of torus allows nodes located near the boundary to have the same number of connections probabilistically as a node located near the centre. Therefore the consideration of a torus implies that there is no need to consider special cases occurring near the boundary of the region. This often simplifies the analysis without obscuring the relationship between main performance-impacting parameters.

\subsection{Wireless communication model}
We consider two types of wireless networks: the cellular network and the complementary network. It is assumed that every node is directly connected to at least one cellular base station (BS) at any time. 
In the complementary network, devices communicate with one another via ad hoc connections using device-to-device communication technologies such as Bluetooth, WiFi or DSRC \cite{DSRC-250}. These ad hoc connections are  usually of high capacity, and can be exploited by mobile devices to cooperatively share content of common interest.

Due to the limited communication range, the ad hoc connection between two nodes only emerges opportunistically when they move close to each other. Considering a commonly used model, called the \emph{unit disk connection model} (UDM), two nodes are directly connected iff their Euclidean distance is not larger than the radio range $r_0$. 
Adopting this commonly used connection model, we say that two nodes \emph{meet} each other when they move into the radio range of each other. 
Consequently define the \emph{inter-meeting time} of two nodes as the time interval between two consecutive meetings of two nodes (a more rigorous definition is given later in Section \ref{section_inter_meeting}).

There have been a number of studies on the distribution of the inter-meeting time. 
In particular, Cai and Eun \cite{cai-crossing-2009} analytically proved that for two nodes moving in a finite area (with reflective or wrapping boundary) under random waypoint or random walk mobility models, their inter-meeting time has an exponential distribution, whereas the inter-meeting time of nodes moving in an infinite area follows a power-law distribution. 

This paper also considers that the inter-meeting time follows an exponential distribution. The analysis presented in this paper however is directly applicable to other inter-meeting time distributions as shown in Section \ref{section_inter_meeting}.

\subsection{Cooperative content dissemination strategy}\label{section_strategy}
This subsection describes the cooperative content dissemination and offloading strategy.

Consider that a content provider has $M$ messages to deliver to the $N$ nodes. The content dissemination process has three phases: initial phase, sharing phase and complement phase.

\begin{enumerate}[{Phase} 1]
\item Initially at time $t=0$, the cellular BSs transmit  $\beta$ packets to $\beta$ different nodes via cellular networks. Note that the content in these packets depends on the coding scheme to be described in Section \ref{section_coding_scheme}.

\item Then the network enters into the \emph{sharing phase}, where the nodes broadcast their received packets via ad hoc connections using the Susceptible-Infected-Recovered (SIR) epidemic scheme, to be described in Section \ref{section_SIR}.

\item At time $\mathcal{T}_{end}$ (called deadline), the sharing phase stops and then the network enters into the \emph{complement phase}, where every node requests  the remaining packets required to decode all $M$ messages from BSs, where the \emph{remaining packets} required to decode the messages are determined by the coding scheme.
\end{enumerate}

The main objective of the content dissemination strategy is to minimise the total number of packets requiring to be transmitted through cellular networks, which include the packets transmitted in the initial phase and in the complement phase, while guaranteeing the successful delivery of the content.

\subsubsection{Coding scheme}\label{section_coding_scheme}

We consider the use of a simple erasure coding scheme. Specifically, the $M$ messages are encoded into $\beta$ \emph{coded packets}  \cite{sathiamoorthy-distributed-2013}. We assume that the Galois field \cite{sathiamoorthy-distributed-2013,li-coding-2013} used in the encoding process is large enough so that the content provider can generate $\beta$ linear-independently coded packets.  It follows that each node needs to acquire $M$ \emph{distinct} coded packets to reconstruct all the $M$ messages.

Consequently, if a node receives $B<M$ distinct coded packets by the end of the sharing phase, then it needs to request $M-B$ coded packets in the complement phase via cellular networks in order to decode all $M$ messages.

For benchmarking, we also consider the case that no coding technique is employed. In this case, the BSs transmit $\hat{\beta}_m \in\{0,1,\dots,\beta\}$ copes of the $m^{th}$ message in the first phase for $m=\{1,2,\dots, M\}$ and $\sum_{m=1}^M \hat{\beta}_m=\beta$. Then each node needs to receive at least one copy of each message.

\subsubsection{Epidemic sharing scheme}\label{section_SIR}
In the sharing phase, the $\beta$ packets are shared among $N$ nodes using a Susceptible-Infected-Recovered (SIR) epidemic sharing scheme. 

Without loss of generality, consider the epidemic sharing of an arbitrary packet, say packet $j$. Using a classic SIR broadcast scheme, a node in the network can be in any of the following three states: the node that has never received the packet $j$ is in the state of susceptible ($S_j$). A susceptible node goes into the state of infected and infectious ($I_j$) immediately after receiving the packet $j$. The node in state $I_j$ keeps transmitting the packet $j$ to every node it meets for a certain time period, which is referred to as the \emph{active period}. Denote by $\tau_h$ the length of the active period of a type-h node. Note that $\tau_h$ is a pre-determined value, which is the same for all nodes of the same type. After the active period, the node recovers and enters into state $R_j$. A recovered node stops transmitting and receiving the packet $j$. The nodes that have received the packet $j$ are referred to as the \emph{informed nodes of the packet $j$}.

Note that the value of $\tau_h$ may be different for different types of nodes. For example, a pre-installed WiFi AP can have a significantly larger value of $\tau_h$ compared with other mobile devices that are powered by battery. Further, the value of $\tau_h$ for a mobile device can be tuned by introducing some incentives (e.g. a lower subscription fee or some rewards) to the mobile users \cite{bo-mobile-2012}, so that the mobile users are willing to share more packets with other users.

When the length of the sharing phase $\mathcal{T}_{end}$ is long, the epidemic sharing process stops naturally (i.e. reaches the \emph{steady state}) when, for all the packets, there is no infectious node. We are particularly interested in the case where the $\mathcal{T}_{end}$ is long because it is suitable for delay-tolerant content disseminations and minimises the traffic load of the cellular network by fully utilising the complementary networks to share content.

Note that this paper considers a large network where $N\gg M$. Furthermore, in the next section, when we consider an \emph{asymptotic} network with $N\rightarrow\infty$, we increase the network area $L\rightarrow\infty$ while keeping the density of every type of nodes unchanged. That is, a well-known extended network model is considered. The analytical results obtained are therefore applicable for a large network only.

\section{Analysis of the content dissemination process}\label{section_analysis}
The main challenge in the analysis of the content dissemination process is the characterisation of the packet propagation process in the sharing phase. This section first analyses the propagation process of a single packet, say the packet $j$, and then generalises to multiple packets.

\subsection{Characterising the ad hoc connections}\label{section_inter_meeting}
Denote by $T_{h,k}$ the inter-meeting time between a randomly-chosen type-h node and a randomly-chosen type-k node, for $h,k\in\{1,2,\dots,H\}$. Assume that $T_{h,k}$ follows an exponential distribution with mean $\lambda_{h,k}$.  The pdf of the inter-meeting time is then given by $\Pr(T_{h,k}=t)=\lambda_{h,k}\exp(-\lambda_{h,k} t)$. It follows that the probability that the type-h node meets the type-k node during the active period $\tau_h$ of the type-h node is
\begin{eqnarray}\label{eq_contact_prob}
\gamma_{h,k}=\int_0^{\tau_h} \Pr(T_{h,k}=t) dt
= 1-\exp(-\lambda_{h,k} \tau_h).
\end{eqnarray}

Note that given other probability distribution of inter-meeting time, one can use a similar method to calculate $\gamma_{h,k}$. Thus our analysis does not critically depend on the assumption of the exponential inter-meeting time distribution.

\subsection{Extinct probability}
In this subsection, we construct a multi-type branching process \cite{athreya-branching-1972,harris-the-2002} to study the number of informed nodes of a  typical packet, say packet $j$, where the $0^{th}$ generation of the branching process includes the nodes that receive the packet $j$ at time $0$ (i.e. in the first phase). Further, the number of type-k children generated by a type-h node is denoted by a random variable $\hat  Q_{h,k}$. 
Because nodes move independently of one another, it is evident that $\hat Q_{h,k}$ follows a Binomial distribution $\mathrm{Bin}(N_k,\gamma_{h,k})$, where $\gamma_{h,k}$ is given in Eq.\ \ref{eq_contact_prob}. Denote by $\mathcal{X}_\alpha^h$ the number of type-h nodes in the $\alpha^{th}$ generation. The branching process modelling the number of informed nodes for a typical packet becomes \emph{extinct} if there exists an integer value of $\alpha$ such that $\sum_{h=1}^H \mathcal{X}_\alpha^h =0$.

The following result is required in the later analysis:
\begin{lem}[Threshold phenomenon]
Define matrix $\hat\Xi\triangleq\{\mathbb{E}[\hat Q_{h,k}]=N_k\gamma_{h,k}; h,k=1,2,...,H\}$. Let $R_q$ be the largest eigenvalue of $\hat\Xi$. Then the branching process will become extinct with probability 1 if and only if $R_q\leq 1$.
\end{lem}
This result can be readily obtained by applying Theorem 2 in \cite[Chapter V]{athreya-branching-1972}. Hence the proof is omitted.

When $R_q > 1$, there is a positive probability that the branching process does not become extinct, i.e. the packet can be disseminated to a significant fraction of nodes as the network size becomes asymptotically large. 

Note that in a heterogeneous network, the type of source nodes that a branching process is rooted at can have a significant impact on the probability that the branching process becomes extinct. Denote by $w_h$ the \emph{extinct probability of type-h source}, which is defined as the probability that a branching process rooted at a type-h node becomes extinct. The following theorem characterises $w_h$.
\begin{thm}[Extinct probability]\label{thm_extinct}
Consider an asymptotic network with $L\rightarrow\infty$ while keeping nodes' density unchanged and the node's communication range varying with $L$ such that $N_k\gamma_{h,k}$ is a finite constant \footnote{This condition is required to avoid triviality in the analysis}. The extinct probabilities $w_h$ for $h=1,2,\dots,H$ are the solutions to the following system of equations:
\begin{equation}\label{eq_extinct_system}
w_h=  \exp\left( \sum_{k=1}^H N_k\gamma_{h,k}(w_k-1) \right), ~\mathrm{ for }~ h=1,2,\dots,H.
\end{equation}
\end{thm}

\begin{proof}
Firstly, note that in the construction of the branching process, we consider that every node that a type-h node meets is a susceptible node. As the packet propagates, the probability that a type-h node meets an informed node increases. This may reduce the expected number of newly informed nodes generated by an infectious node. However, there is no need to consider the impact of this effect in the analysis of the extinct probability because when analysing the extinct probability, we are only interested in the case that the fraction of recipients is vanishingly small (i.e. becomes extinct) as $L\rightarrow\infty$ and $N\rightarrow\infty$. Accordingly, the probability that a type-h node meets an informed node is vanishingly small and hence negligible.

Because $N_k\gamma_{h,k}$ is a finite constant as the network size increases, the distribution of $\hat Q_{h,k}$, i.e. a Binomial distribution $\mathrm{Bin}(N_k,\gamma_{h,k})$, approaches a Poisson distribution with an expected value $N_k\gamma_{h,k}$ \cite{pollard-poisson-1997}. The difference between the $\hat Q_{h,k}$ and its Poisson distribution counterpart (denoted by $Q_{h,k}$) diminishes as $N_k\rightarrow\infty$ and $\gamma_{h,k}\rightarrow 0$, where the convergence rate is given in \cite{pollard-poisson-1997}.

Denote by $G_{h,k}(s)$ the probability generating function of $Q_{h,k}$:
\begin{equation}
G_{h,k}(s)=\mathbb{E}[s^{Q_{h,k}}]
=\exp\left( N_k\gamma_{h,k}(s-1) \right).
\end{equation}

Further, define the multi-variate probability generating function $\hat{G}_h(\mathbf{s})\triangleq \mathbb{E}[s_1^{Q_{h1}}s_2^{Q_{h2}} \dots s_H^{Q_{hH}}]$, where $\mathbf{s}\triangleq \{s_1, s_2, \dots, s_H\}$ is a row vector. It can be shown that
\begin{equation}
\hat{G}_h(\mathbf{s})= G_{h,1}(s_1)G_{h,2}(s_2)\dots G_{h,H}(s_H).
\end{equation}

Denote the extinct probabilities by a row vector $\mathbf{w}\triangleq [w_1, w_2, \dots, w_H]$. Then according to Theorem 2 in \cite[Chapter V]{athreya-branching-1972}, the extinct probabilities satisfy $\mathbf{w}=\mathbf{\hat G}(\mathbf{w})$, where $\mathbf{\hat G}(\mathbf{s})$ is a row vector $\left[ \hat{G}_1(\mathbf{s}),\hat{G}_2(\mathbf{s}),\dots , \hat{G}_H(\mathbf{s})\right]$.
The conclusion follows that the extinct probabilities $w_h$ for $h=\{1,2,\dots,H\}$ are the solutions to the following system of equations:
\begin{eqnarray}
w_1&=& \prod_{k=1}^H \exp\left( N_k\gamma_{1,k}(w_k-1) \right), \\
&\dots&  \nonumber\\
w_H&=& \prod_{k=1}^H \exp\left( N_k\gamma_{H,k}(w_k-1) \right) ,\nonumber
\end{eqnarray}
where $N_k$ is the number of type-k nodes in the network and $\gamma_{h,k}$ is given by Eq. \ref{eq_contact_prob}.
\end{proof}

Using Theorem \ref{thm_extinct}, we can further obtain the extinct probability for several special-case of networks.

\begin{corr}\label{corr_extinct_single_type}[Extinct probability for homogeneous networks]
Consider the special case of a network with only $H=1$ type of nodes, the extinct probability is $w_1=-\frac{\mathcal{W}(-N\gamma_{1,1}\exp(-N\gamma_{1,1}))}{N\gamma_{1,1}}$, where $\mathcal{W}(.)$ is the Lambert-W function.
\end{corr}

\begin{corr}\label{corr_extinct_multiple}[Extinct probability with multiple source nodes]
Suppose that a packet is initially broadcast from $\beta=\sum_{h=1}^H \beta_h$ source nodes at the beginning of phase 2, where $\beta_h$ is the number of type-h source nodes. Then the extinct probability for this packet is
$
\prod_{h=1}^H w_h^{\beta_h}
$.
\end{corr}

The proof of the above two corollaries is straightforward and hence omitted.

Note that when the branching process does \emph{not} become extinct, the packet is disseminated to a significant number of nodes and we say that \emph{the packet spreads out}. The next sub-section quantifies the number of recipients of a packet when it spreads out.

\subsection{Expected fraction of recipients}

Denote by $\hat{z}_h$ the expected fraction of type-h nodes which receive the packet $j$ in the steady state, where the packet propagation starts from a randomly-chosen source node \emph{and} the packet spreads out. 
Now we further investigate $\hat{z}_h$.

\begin{thm}[Fraction of recipients when a packet spreads out]\label{thm_size}
In a large network with $N\rightarrow\infty$, suppose that the packet $j$ is broadcast from a randomly chosen source node, no matter which type the source node belongs to. Given that the packet is spread out, the expected fractions of recipients $\hat{z}_h$ for $h=1,2,\dots,H$ are the solutions to the following system of equations
\begin{equation}\label{eq_size_system}
1-\hat{z}_h = \exp\left( -\sum_{k=1}^H N_k\gamma_{k,h}\hat{z}_k \right), ~\mathrm{ for }~ h=1,2,\dots,H.
\end{equation}
\end{thm}

This theorem can be readily obtained from the analysis of epidemics\cite{ball-the-1993,britton-stochastic-2010} and the proof is omitted. 

Next we consider the case that there are more than one source node of the packet $j$.

\begin{corr}[Fraction of recipients with multiple source nodes]\label{corr_size}
Suppose that initially at time $0$, there are $\sum_{h=1}^H \beta_h=\beta$ nodes that have received the packet $j$, where $\beta_h$ is the number of type-h source nodes. Denote by $z(\beta_1,\dots,\beta_H)$ the expected fraction of nodes, out of the total $N$ nodes, which receive the packet in the steady state. In a large network with $N\rightarrow\infty$, there holds 
\begin{eqnarray}
z(\beta_1,\dots,\beta_H)=\left( \sum_{h=1}^H \frac{N_h  \hat{z}_h}{N}\right) 
\left( 1-\prod_{h=1}^H w_h^{\beta_h} \right).
\end{eqnarray}
\end{corr}

\begin{proof}
From Corollary \ref{corr_extinct_multiple}, the probability that the packet spreads out is $\left( 1-\prod_{h=1}^H w_h^{\beta_h} \right)$. Note that if the packet is not spread out, the fraction of recipients goes to 0 as $N\rightarrow\infty$. If the packet spreads out, the expected number of type-h recipients is $N_h \hat{z}_h$, where $\hat{z}_h$ is given by Theorem \ref{thm_size}. The conclusion follows.
\end{proof}

For the special case that there is only $H=1$ type of node in a network, a closed form expression can be obtained. 
\begin{corr}[Fraction of recipients for homogeneous networks]\label{corr_frac_H1}
Suppose that there is only $H=1$ type of node in a network and a packet is sent to $\beta$ different nodes in the first phase. Then in the steady state, the expected fraction of recipients of this packet is 
\begin{eqnarray}\label{eq_frac_H1}
z(\beta)&=&\left( 1+ \frac{\mathcal{W}(-N\gamma_{1,1}\exp(-N\gamma_{1,1}))}{N\gamma_{1,1}}\right) \nonumber\\
&&\times \left( 1- \left(-\frac{\mathcal{W}(-N\gamma_{1,1}\exp(-N\gamma_{1,1}))}{N\gamma_{1,1}}\right)^\beta\right).
\end{eqnarray}
\end{corr}

\section{Minimising the traffic load of cellular networks}\label{section_optimization}
Based on the above characterisation of the content dissemination process, this section investigates the optimal content dissemination strategy that minimises the total traffic load of cellular networks.

Recall that in the initial phase, the BSs transmit $\beta$ packets to $\beta$ different nodes through the cellular network. 
Denote by $Y$ the expected number of packets that BSs need to transmit in the complement phase.

\begin{defn}[Cellular traffic load]
The \emph{cellular traffic load} is the expected number of packets transmitted by BSs through the cellular network, which consist of the packets transmitted in the first and the third phases, i.e. $\beta+Y$, in order to transmit $M$ messages to $N$ nodes.
\end{defn}

Note that the value of $\beta$ determines the value of $Y$, which is calculated later in this section. Then the problem of minimising the cellular traffic load can be formulated as follows:
\begin{equation}\label{eq_optimization_general}
\begin{aligned}
& \underset{\beta}{\text{Minimise}}
& & \beta+Y \\
& \text{Subject to}
& & \beta\in\{1,2,\dots,N \}.
\end{aligned}
\end{equation}

Using the erasure coding technique introduced in Section \ref{section_coding_scheme}, the BSs push $\beta$ coded packets to $\beta$ different nodes in the initial phase. The following lemma gives the optimum strategy to choose the source nodes to disseminate the packets.

\begin{lem}[Sharing maximisation strategy]\label{lem_packet_allocation2}
Label all nodes in the network in the descending order of their values of $w_h$, which is given by Theorem \ref{thm_extinct}. If more than one node have the same value of $w_h$, their order can be arbitrarily assigned. Suppose that BSs push $\beta$ encoded packets to $\beta$ different nodes in the initial phase. Then the optimal strategy that minimises the cellular traffic load is to push $\beta$ \emph{different} packets to the \emph{first} $\beta$ nodes in the above order.
\end{lem}

\begin{proof}
First, it can be readily shown that the strategy that minimises the cellular traffic load is to push $\beta$ \emph{different} coded packets in the initial phase rather than pushing multiple copies of any coded packet. 

Because different encoded packets are shared independently of one another, we next consider a randomly-chosen packet, say packet $j$. It is obvious that to minimise the cellular traffic load, one needs to maximise the number of recipients of packet $j$. 
According to Corollary \ref{corr_size}, the expected fraction of nodes, out of the total $N$ nodes, which receive packet $j$ in the steady state is 
\begin{eqnarray}
z(\beta_1,\dots,\beta_H)=\left( \sum_{h=1}^H \frac{N_h  \hat{z}_h}{N}\right) 
\left( 1-\prod_{h=1}^H w_h^{\beta_h} \right),
\end{eqnarray}
where $\beta_1+\beta_2+\dots+\beta_H=1$ because we are now considering the propagation of single packet - packet $j$. In other words, only one value among $\beta_1,\beta_2,\dots,\beta_H$ is equal to one and the other values are all equal to 0. It is obvious that to maximise $z(\beta_1,\dots,\beta_H)$, one should assign the values of $\beta_1,\beta_2,\dots,\beta_H$ in a way that minimises $\prod_{h=1}^H w_h^{\beta_h}$, i.e. let the only non-zero value $\beta_k=1$ for the $k^{th}$ type of nodes that have the smallest value of the extinct probability $w_k$ among all $w_h$ for $h\in\{1,2,\dots,H\}$.

It can be shown that to maximise the sharing performance, the $\beta$ coded packets need to be pushed to $\beta$ \emph{different} nodes. Therefore, when the number of type-k nodes $N_k$ is less than the total number of packets $\beta$, some packets need to be pushed to the nodes that have the second (and if needed, the third, forth, etc.) smallest value of the extinct probability $w_i$ among all $w_h$ for $h\in\{1,2,\dots,H\}$.
\end{proof}


Next we focus on determining the optimum value of $\beta$ for the special case of a homogeneous network with $H=1$. The optimum value of $\beta$ for the more general case of a heterogeneous network with $H>1$ can be determined analogously albeit with greater complexity.

Denote by random variable $B$ the number of packets received by a randomly-chosen node at the end of the sharing phase. Then $B$ follows a Binomial distribution, i.e. the probability that a node receives $B=b$ packets is $\Pr(B=b)=\binom{\beta}{b} (z(1))^b (1-z(1))^{\beta-b}$, where $\binom{\beta}{b}=\frac{\beta!}{b!(\beta-b)!}$.

Then in the complement phase, the number of packets that need to be transmitted to a randomly-chosen node is $(M-B)^+$, where $(x)^+=\max\{0,x\}$.

Finally, the expected number of packets that the BSs need to transmit in the complement phases is
\begin{eqnarray}
Y&=&N\mathbb{E} [ (M-B)^+ ]\\
&=&N\sum_{b=0}^M (M-b) \binom{\beta}{b} (z(1))^b (1-z(1))^{\beta-b}.
\end{eqnarray}

Then the optimisation problem in Eq. \ref{eq_optimization_general} becomes
\begin{equation}
\begin{aligned}
& \underset{\beta}{\text{Minimise}}
& & \beta+N\sum_{b=0}^M (M-b) \binom{\beta}{b} z_1^b (1-z_1)^{\beta-b} \\
& \text{Subject to}
& & \beta\in\{1,2,\dots,N \}.
\end{aligned}
\end{equation}

This optimisation problem  can be readily solved numerically using Matlab, where the results are presented in the next section.

\section{Simulation and discussion}\label{section_simulation}
This section reports on simulations to verify the accuracy of the analysis presented in the previous sections. 
The simulations are conducted using a mobile network simulator written in C++. 
%
%
Specifically, $N=960$ nodes are uniformly deployed on a square $(0,8000]^2$. Consequently the nodes' density equals to $15$ nodes/$km^2$, which is equal to the density of cabs in New York metropolitan area \cite{taxi_density}.  After initial deployment of the nodes, they start to move according to the random direction mobility model. When the nodes hit the boundary of the square, they may appear from the other side of the boundary.  It can be shown that when nodes move according to the random direction mobility model, the inter-meeting time follows an exponential distribution. 
The node's speed is $V=10m/s$ (typical vehicle moving speed). The radio range $r_0=20$m or $250$m (typical radio ranges using Wi-Fi Tethering  or DSRC \cite{DSRC-250}). 
Every point shown in the simulation result is the average value from 500 simulations.

Consider two types of nodes moving according to the random direction model with speeds $V_1=10m/s$ and $V_2=0$, i.e. mobile and static nodes, and there are equal number of nodes in each type. Fig. \ref{pic_change_tau_half_all} shows the probability that a packet spreads out and the expected fraction of recipients of a single packet.  The analytical result of the probability that a packet spreads out when the source node is of type-1 (resp. type-2) is given by $w_1$ (resp. $w_2$) from Theorem \ref{thm_extinct}. The analytical result of the expected fraction of recipients of a packet when the source node is of type-1 (resp. type-2) is given by $z(1,0)$ (resp. $z(0,1)$) from Corollary \ref{corr_size}. It is interesting to note that the probability that a packet spreads out and the expected fraction of recipients can be significantly affected by the type of source node. More specifically, other things being equal, it can be seen that a mobile source node can spread the packet to more recipients than a static source node.

\begin{figure}[hbt]
\centering
\includegraphics[scale=0.60]{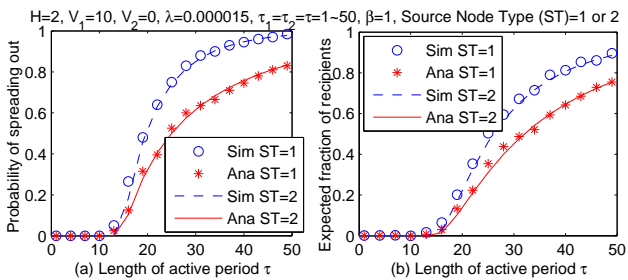}\vspace{-3px}
\caption{Simulation and analytical results of (a) The probability that a packet spreads out and (b) the expected fraction of recipients of a packet.}
\label{pic_change_tau_half_all}\vspace{-4px}
\end{figure}

Fig. \ref{pic_plot_wifi_frac} shows the results of another interesting case where the complementary network consists of a set of fixed-location base stations (e.g. WiFi APs). In the second phase, the message is disseminated from these WiFi APs to mobile users (i.e. type-1 nodes). 
Specifically, we consider that a small number of WiFi APs are randomly and uniformly deployed in a given area, which is a widely-used setting for AP deployment \cite{andrews-seven-2013}. 
Because the WiFi APs are usually connected to the Internet via wired connections, we set $\gamma_{2,2}=1$. Other parameters are the same as those in the previous sub-section. 
A message is transmitted to all the WiFi APs at time 0. Then the WiFi APs keep transmitting the message for a given time period $\tau_2$. In Fig. \ref{pic_plot_wifi_frac}, we let the active period of type-2 nodes be $\tau_2=500, 1500$ while varying the active period of type-1 nodes. 
Note that $\tau_1=0$ corresponds to the traditional case \cite{lee-mobile-2013,bulut-wifi-2013} where nodes do not cooperatively share received packets and they solely rely on WiFi APs to offload data traffic from cellular networks. It can be seen in Fig. \ref{pic_plot_wifi_frac} that a longer active period of mobile nodes $\tau_1$ leads to a larger expected fraction of recipients. It is obvious that packet sharing using ad hoc connections between mobile nodes can significantly increase the number of recipients of a packet, hence reducing the number of transmissions required by BSs.

\begin{figure}[hbt]
\centering
\includegraphics[scale=0.55]{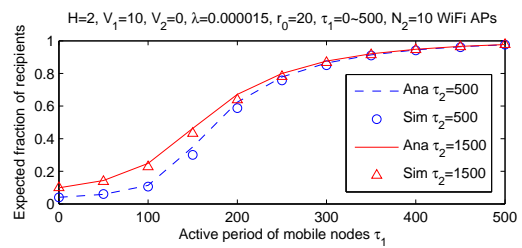}\vspace{-3px}
\caption{The expected fraction of recipients of a packet in a network with 10 fixed WiFi APs and some mobile nodes.}
\label{pic_plot_wifi_frac}\vspace{-4px}
\end{figure}

The following results further  evaluate the cellular traffic loads.

Fig. \ref{pic_single_packet_optimization} shows the expected cellular data traffic load $\beta+Y$ with different values of $\beta$, viz. the number of packets sent in the first phase, using different content dissemination strategies. To study the impact of coding on the cellular traffic load, we consider two networks with the same setting except that one network employs the erasure coding technique (c.f. Section \ref{section_strategy}) but the other network does not. The performance of the network without employing network coding can be readily obtained using the same technique adopted for analysing network employing the erasure coding technique. 

Several interesting trends can be observed in Fig. \ref{pic_single_packet_optimization}. 
Firstly, when $\beta$ is small, two networks have a similar and relatively high cellular traffic load. This is because only a limited number of nodes receive the packets through the complementary network, hence most packets are directly transmitted to the users via cellular networks. 
As $\beta$ increases, the cellular traffic load first decreases rapidly,  due to a rapid increase in the expected fraction of recipients in the sharing phase. Then after a certain point, the cellular traffic load starts to gradually increases as $\beta$ further increases. This is because the expected fraction of recipients has limited increase when $\beta$ increases further; on the other hand the increase in $\beta$ causes more cellular data traffic. It is interesting to note that sending out more packets in the initial phase is not always beneficial.

Furthermore, it is interesting to note that above a certain value of $\beta$, e.g. $\beta=17$ in Fig. \ref{pic_single_packet_optimization}(a), the traffic load of cellular networks employing coding is significantly smaller than that of networks without coding, due to the following reason. Recall that when $\beta>M$, each node only needs $M$ different coded packets to decode all $M$ messages when coding is employed. On the other hand, in a network without coding, there is a non-zero probability that two packets received at a node contain the same message. Therefore a node may need more than $M$ packets. 
%

Note that without employing the cooperative content dissemination strategy, the nodes need to request all packets from the cellular network and the traffic load of cellular networks is $960$. Compared with the values of $\beta+Y$ in Fig. \ref{pic_single_packet_optimization}, it is evident that the cooperative content dissemination strategy can significantly reduce the traffic load of cellular networks.

\begin{figure}[hbt]
\centering
\includegraphics[scale=0.55]{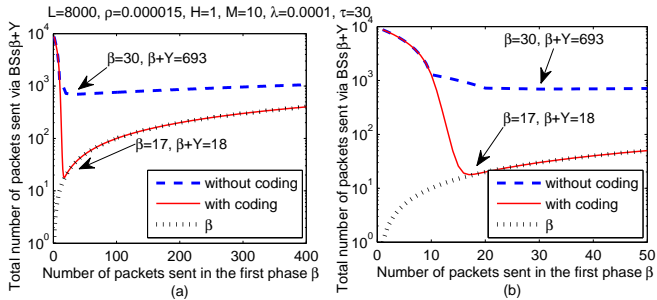}\vspace{-3px}
\caption{Comparison between the expected cellular data traffic loads $\beta+Y$ in networks with and without coding. Note that subplot (b) shows the range $\beta\in[1,50]$ of subplot (a).}
\label{pic_single_packet_optimization}\vspace{-4px}
\end{figure}

\section{Conclusion and future work}\label{section_conclusion}
This paper investigated a cooperative content dissemination strategy for heterogeneous networks consisting of different types of devices. The content dissemination strategy can effectively offload a significant amount of data traffic from cellular networks to complementary networks such as WiFi and device-to-device networks using ad hoc connections that emerge when devices move and meet one another. Theoretical analysis for the content dissemination process was presented. On that basis, the optimal design of the content dissemination strategy was discussed, which maximally reduces traffic load of cellular networks while guaranteeing the successful delivery of all content.

In our paper, we consider that the duration of the sharing phase $\mathcal{T}_{end}$ is sufficiently long such that the epidemic sharing process is able to reach its steady state. In the future, one may consider the case where only a short time period is allowed for the sharing phase, causing the epidemic sharing process to terminate before it reaches the steady state. In this case, non-trivial analysis is required to calculate the number of recipients of a packet at an arbitrary time instant. Furthermore, it is an interesting extension of our work  to consider different probability distributions for the inter-meeting time of nodes, which can be affected by the nodes' mobility and network area as described in Section \ref{section_system_model}.

\bibliographystyle{IEEEtran}

\begin{thebibliography}{10}
\providecommand{\url}[1]{#1}
\csname url@samestyle\endcsname
\providecommand{\newblock}{\relax}
\providecommand{\bibinfo}[2]{#2}
\providecommand{\BIBentrySTDinterwordspacing}{\spaceskip=0pt\relax}
\providecommand{\BIBentryALTinterwordstretchfactor}{4}
\providecommand{\BIBentryALTinterwordspacing}{\spaceskip=\fontdimen2\font plus
\BIBentryALTinterwordstretchfactor\fontdimen3\font minus
  \fontdimen4\font\relax}
\providecommand{\BIBforeignlanguage}[2]{{%
\expandafter\ifx\csname l@#1\endcsname\relax
\typeout{** WARNING: IEEEtran.bst: No hyphenation pattern has been}%
\typeout{** loaded for the language `#1'. Using the pattern for}%
\typeout{** the default language instead.}%
\else
\language=\csname l@#1\endcsname
\fi
#2}}
\providecommand{\BIBdecl}{\relax}
\BIBdecl

\bibitem{li-multiple-2013}
Y.~Li, M.~Qian, D.~Jin, P.~Hui, Z.~Wang, and S.~Chen, ``Multiple mobile data
  offloading through disruption tolerant networks,'' \emph{IEEE Transactions on
  Mobile Computing}, vol.~13, no.~7, pp. 1579 - 1596, 2014.
  
  \bibitem{andrews-seven-2013}
J.~G. Andrews, ``Seven ways that HetNet are a cellular paradigm shift,''
  \emph{IEEE Communications Magazine}, vol.~51, no.~3, pp. 136--144, 2013.
  
  \bibitem{Kannan08Robust}
A.~Kannan, B.~Fidan and G.~Mao, ``Robust Distributed Sensor Network Localization Based on Analysis of Flip Ambiguities,''
  \emph{IEEE Globecom}, pp. 1--6, 2008.

  
\bibitem{Ge15Energy}
X.~Ge, S.~Tu, T.~Han, Q.~Li and G.~Mao, ``Energy Efficiency of Small Cell Backhaul Networks Based on Gauss-Markov Mobile Models,''
  \emph{IET Networks}, vol.~4, no.~2, pp. 158--167, 2015.

 \bibitem{Mao09Road}
R.~Mao and G.~Mao, ``Road Traffic Density Estimation in Vehicular Networks,''
  \emph{IEEE WCNC}, pp. 4700--4705, 2013.

\bibitem{Mao09Graph}
G.~Mao and B.~D.~O.~Anderson, ``Graph Theoretic Models and Tools for the Analysis of Dynamic Wireless Multihop Networks,''
  \emph{IEEE WCNC}, pp. 1--6, 2009.
  
\bibitem{Mao06Online}
G.~Mao, B.~D.~O.~Anderson and B.~Fidan, ``Online calibration of path loss exponent in wireless sensor networks,''
  \emph{IEEE Globecom}, pp. 1--6, 2006.
  
\bibitem{zhang-enabling-2013}
Z.~Da and Y.~Chai~Kiat, ``Enabling efficient wifi-based vehicular content
  distribution,'' \emph{IEEE Transactions on Parallel and Distributed Systems},
  vol.~24, no.~3, pp. 479--492, 2013.

\bibitem{lee-mobile-2013}
K.~Lee, J.~Lee, Y.~Yi, I.~Rhee, and S.~Chong, ``Mobile data offloading: How
  much can wifi deliver?'' \emph{IEEE/ACM Transactions on Networking}, vol.~21,
  no.~2, pp. 536--550, 2013.

\bibitem{bulut-wifi-2013}
E.~Bulut and B.~K. Szymanski, ``Wifi access point deployment for efficient
  mobile data offloading,'' \emph{SIGMOBILE Mob. Comput. Commun. Rev.},
  vol.~17, no.~1, pp. 71--78, 2013.

\bibitem{aijaz-a-2013}
A.~Aijaz, H.~Aghvami, and M.~Amani, ``A survey on mobile data offloading:
  technical and business perspectives,'' \emph{IEEE Wireless Communications},
  vol.~20, no.~2, pp. 104--112, 2013.

\bibitem{bo-mobile-2012}
H.~Bo, H.~Pan, V.~S.~A. Kumar, M.~V. Marathe, S.~Jianhua, and A.~Srinivasan,
  ``Mobile data offloading through opportunistic communications and social
  participation,'' \emph{IEEE Transactions on Mobile Computing}, vol.~11,
  no.~5, pp. 821--834, 2012.

\bibitem{niyato-a-2010}
D.~Niyato and E.~Hossain, ``A unified framework for optimal wireless access for
  data streaming over vehicle-to-roadside communications,'' \emph{IEEE
  Transactions on Vehicular Technology}, vol.~59, no.~6, pp. 3025--3035, 2010.

\bibitem{panichpapiboon-a-2012}
S.~Panichpapiboon and W.~Pattara-Atikom, ``A review of information
  dissemination protocols for vehicular ad hoc networks,'' \emph{IEEE
  Communications Surveys and Tutorials}, vol.~14, no.~3, pp. 784--798, 2012.

\bibitem{franceschetti-random-2007}
M.~Franceschetti and R.~Meester, \emph{Random networks for communication: From
  Statistical Physics to Information Systems}, Cambridge University Press, 2007.

\bibitem{camp-a-2002}
T.~Camp, J.~Boleng, and V.~Davies, ``A survey of mobility models for ad hoc
  network research,'' \emph{Wireless Communications and Mobile Computing},
  vol.~2, no.~5, pp. 483 -- 502, 2002.

\bibitem{nain-properties-2005}
P.~Nain, D.~Towsley, B.~Liu, and Z.~Liu, ``Properties of random direction
  models,'' in \emph{Proceedings IEEE INFOCOM}, vol.~3, 2005, pp. 1897-- 1907.

\bibitem{DSRC-250}
J.~Rezgui and S.~Cherkaoui, ``About deterministic and non-deterministic
  vehicular communications over DSRC/802.11p,'' \emph{Wireless Communications
  and Mobile Computing}, vol.~14, no.~15, pp. 1435-1449, 2012.

\bibitem{cai-crossing-2009}
H.~Cai and D.~Y. Eun, ``Crossing over the bounded domain: From exponential to
  power-law intermeeting time in mobile ad hoc networks,'' \emph{IEEE/ACM
  Transactions on Networking}, vol.~17, no.~5, pp. 1578--1591, 2009.

\bibitem{sathiamoorthy-distributed-2013}
M.~Sathiamoorthy, A.~Dimakis, B.~Krishnamachari, and F.~Bai, ``Distributed
  storage codes reduce latency in vehicular networks,'' \emph{IEEE Transactions
  on Mobile Computing}, vol.~13, no.~9, pp. 2016-2027, 2013.

\bibitem{li-coding-2013}
Y.~Li, D.~Jin, Z.~Wang, L.~Zeng, and S.~Chen, ``Coding or not: Optimal mobile
  data offloading in opportunistic vehicular networks,'' \emph{IEEE
  Transactions on Intelligent Transportation Systems}, vol.~15, no.~1, pp.
  318-333, 2013.

\bibitem{athreya-branching-1972}
K.~B. Athreya and P.~E. Ney, \emph{Branching Processes}, Springer-Verlag, 1972.

\bibitem{harris-the-2002}
T.~E. Harris, \emph{The Theory of Branching Processes}, Dover Publications, 2002.

\bibitem{pollard-poisson-1997}
D.~Pollard, \emph{Poisson Approximation},
  1997.

\bibitem{ball-the-1993}
F.~Ball and D.~Clancy, ``The final size and severity of a generalised
  stochastic multitype epidemic model,'' \emph{Advances in Applied
  Probability}, vol.~25, no.~4, pp. 721--736, 1993.

\bibitem{britton-stochastic-2010}
T.~Britton, ``Stochastic epidemic models: A survey,'' \emph{Mathematical
  Biosciences}, vol. 225, no.~1, pp. 24--35, 2010.

\bibitem{taxi_density}
\BIBentryALTinterwordspacing
MKThink, ``Unsustainable city: density, transportation, and san franciscos
  sustainability,'' 2005. [Online]. Available:
  \url{http://www.mkthink.com/archives/2470}
\BIBentrySTDinterwordspacing

\end{thebibliography}

\end{document}